\documentclass[10pt,journal,twoside,final]{IEEEtran}
\ifCLASSINFOpdf
\else
\fi

\usepackage[cmex10]{amsmath}
\usepackage{amsthm}
\usepackage{amsfonts}
\newtheoremstyle{thm}
  {}
  {}
  {}
  {9pt}
  {\itshape}
  {:}
  {.5em}
  {}

\theoremstyle{thm}
\newtheorem{thm}{Theorem}
\newtheorem{lem}{Lemma}

\newtheoremstyle{pro}
  {}
  {}
  {}
  {18pt}
  {\itshape}
  {:}
  {.5em}
  {}
\theoremstyle{pro}

\newtheoremstyle{thm2}
  {}
  {}
  {}
  {9pt}
  {\itshape}
  {:}
  {\newline}
  {}
\theoremstyle{thm2}
\newtheorem*{rem}{Remarks}
\numberwithin{equation}{section}
\newenvironment{remenum}[1][]{~\vspace*{-\baselineskip}\begin{enumerate}[#1]}{\end{enumerate}}
\newtheoremstyle{thm3}
  {}
  {}
  {}
  {9pt}
  {\itshape}
  {:}
  {9pt}
  {}
\theoremstyle{thm3}
\newtheorem*{singlerem}{Remark}

\usepackage{epsfig}
\usepackage{psfrag}
\usepackage{subfigure}
\usepackage{multirow}
\usepackage{slashbox}
\usepackage{stfloats}
\usepackage{cases}
\begin{document}
\newtheorem{theorem}{Theorem}

\title{Bits Through Deterministic Relay Cascades with Half-Duplex Constraint}
\author{\IEEEauthorblockN{Tobias~Lutz, Christoph~Hausl,~\IEEEmembership{Student Member,~IEEE,} and~Ralf~K\"{o}tter,~\IEEEmembership{Fellow,~IEEE}}
\thanks{Manuscript received June 11, 2009; revised February 26, 2010; accepted June 27, 2011. Date of current version August 22, 2011. This work was supported by the European Commission in the framework of the FP7 (contract number 215252) and by DARPA under the ITMANET program. The material in this paper was presented in part~\cite{LuHaKo08} at the IEEE International Symposium on Information Theory, Toronto, Canada, July 6-11, 2008.

Tobias Lutz and Christoph Hausl are with the Institute for Communications Engineering, TU M\"{u}nchen, 80290 M\"{u}nchen, Germany (Email: \{tobi.lutz, christoph.hausl\}@tum.de).

Ralf K\"{o}tter, deceased, was with the Institute for Communications Engineering, TU M\"{u}nchen, 80290 M\"{u}nchen, Germany.

Communicated by M. Gastpar, Associate Editor for Shannon Theory.}}

\markboth{accepted for publication in IEEE Transactions on Information Theory}
{Lutz \MakeLowercase{\textit{et al.}}: Bits Through Deterministic Relay Cascades with Half-Duplex Constraint}

\maketitle

\begin{abstract}
Consider a relay cascade, i.e. a network where a source node, a sink node and a certain number of intermediate source/relay nodes are arranged on a line and where adjacent node pairs are connected by error-free $(q+1)$-ary pipes. Suppose the source and a subset of the relays wish to communicate independent information to the sink under the condition that each relay in the cascade is half-duplex constrained. A coding scheme is developed which transfers information by an information-dependent allocation of the transmission and reception slots of the relays. The coding scheme requires synchronization on the symbol level through a shared clock. The coding strategy achieves capacity for a single source. Numerical values for the capacity of cascades of various lengths are provided, and the capacities are significantly higher than the rates which are achievable with a predetermined time-sharing approach. If the cascade includes a source and a certain number of relays with their own information, the strategy achieves the cut-set bound when the rates of the relay sources fall below certain thresholds. For cascades composed of an infinite number of half-duplex constrained relays and a single source, we derive an explicit capacity expression. Remarkably, the capacity in bits/use for $q=1$ is equal to the logarithm of the golden ratio, and the capacity for $q=2$ is $1$ bit/use.
\end{abstract}
\begin{IEEEkeywords}\theoremstyle{break}
\newtheorem{bthm}{B-Theorem}
Half-duplex constraint, relay networks, network coding, timing, constrained coding, capacity, capacity region, method of types, golden ratio.
\end{IEEEkeywords}

\IEEEpeerreviewmaketitle

\section{Introduction}
\IEEEPARstart{A} relay cascade is a network where a source node, a sink node and a certain number of intermediate source/relay nodes are arranged on a line. We consider the problem where a source node and certain relay nodes wish to communicate independent messages to the sink under the condition that each relay is half-duplex constrained, i.e. is not able to transmit and receive simultaneously. Throughout the paper, we assume that adjacent node pairs are connected by error-free $(q+1)$-ary pipes. This approach lets us understand half-duplex constrained transmission without having to consider channel noise. Moreover, we may use combinatorial arguments instead of stochastic arguments. 

A natural strategy for half-duplex devices is to define a time-division schedule a priori. Under this assumption, the capacity or rate region of various half-duplex constrained relay channels \cite{Mad05}, \cite{Kho03} and networks \cite{Tou03} has been determined. We will, however, show that predetermined time-sharing falls considerably short of the theoretical optimum or, conversely, higher rates are possible by an information-dependent allocation of the transmission and reception slots of the relays. 

The meaning of information-dependent allocation scheme is illustrated in the following example. Let $\mathcal{W}_0=\{0,\dots,7\}$ be a message set. In each block $i=1,2,\dots$ of length $4$, the source wishes to communicate a randomly chosen message $w_0(i)\in\mathcal{W}_0$ to the destination via a single half-duplex constrained relay node. A direct link between source and destination does not exist. Suppose the alphabet of both source and relay equals $\{0,1,\textrm{N}\}$ where ``N'' indicates a channel use without transmission and $\{0,1\}$ is a $q=2$-ary transmission alphabet. The half-duplex constraint is modeled as follows. When the relay uses symbol~``N'', i.e. the relay is quiet, it is able to listen to the source and otherwise not. Let $\mathbf{x}_0(i)$ be the codeword chosen by the source encoder to represent $w_0(i)$ in block~$i$ and let $\mathbf{x}_1(i)$ indicate the codeword chosen by the relay encoder for representing $w_0(i-1)$ in block~$i$. The coding scheme is illustrated in Table \ref{cod_example_rand}. The source encoder maps each message $w_0(i)$ to $\mathbf{x}_0(i)$ by allocating the corresponding binary representation of $w_0(i)$, i.e. three bits, to four time slots. The precise allocation of the three bits to four time slots is determined by the following protocol. In the first block, the source allocates three bits to the first three time slots of $\mathbf{x}_0(1)$. Now assume that the source has already sent codeword $\mathbf{x}_0(i)$ to the relay. Based on the first two binary digits of the noiselessly received codeword $\mathbf{x}_0(i)$, the relay encoder determines which of the four time slots to use for transmission in $\mathbf{x}_1(i+1)$ according to the following rule: $00$, $01$, $10$, $11$ in $\mathbf{x}_0(i)$ tells the relay to send in the first, the second, the third or the fourth time slot of $\mathbf{x}_1(i+1)$. The binary value to be transmitted in $\mathbf{x}_1(i+1)$ is equal to the third bit in $\mathbf{x}_0(i)$. Since the source encoder knows the scheme used by the relay, it can allocate its three new bits in $\mathbf{x}_0(i+1)$ to those slots in which the relay is able to listen. Hence, the relay encodes a part of its information in the timing of the transmission symbols. The sink estimates message~$w_0(i-1)$ from the received relay codeword~$\mathbf{x}_1(i)$ using both the position of the transmission symbol and its value and obtains~$\hat{w}_0(i)$. In this example, a rate of $0.75$ bit per use is asymptotically achievable if the number of blocks becomes large. By allowing arbitrarily long codewords, we will show that an extension of the strategy approaches $1.1389$~b/u which is also the capacity of the single relay cascade with half-duplex constraint when the transmission alphabet is binary. 
\begin{table}[!t]
\renewcommand{\arraystretch}{1.3}
\caption{The relay encodes a part of the information by the position of the transmission symbols}
\begin{center}
\begin{tabular}{c|| c| c| c| c} \hline
block $i$&$w_0(i)$&  $\mathbf{x}_0(i)$ & $\mathbf{x}_1(i)$ & $\hat{w}_0(i)$ \\ \hline\hline
$\parbox[0pt][1em][c]{0cm}{}i=1$&$1$ ($001$)& $001\textrm{N}$ & $\textrm{NNNN}$ & -\\
$\parbox[0pt][1em][c]{0cm}{}i=2$&$2$ ($010$)& $\textrm{N}010$ & $1\textrm{N}\textrm{N}\textrm{N}$ & $1$ \\ 
$\parbox[0pt][1em][c]{0cm}{}i=3$&$4$ ($100$)& $1\textrm{N}00$ & $\textrm{N}0\textrm{N}\textrm{N}$ & $2$ \\ 
$\parbox[0pt][1em][c]{0cm}{}i=4$&$7$ ($111$)& $11\textrm{N}1$ & $\textrm{N}\textrm{N}0\textrm{N}$ & $4$ \\ 
$\vdots$&$\vdots$&$\vdots$&$\vdots$&$\vdots$\\ \hline
\end{tabular}
\label{cod_example_rand}
\end{center}
\end{table}
The example suggests that information encoding by means of timing is beneficial in the context of half-duplex constrained transmission. A similar example for $q=1$ was shown in \cite{KraMarYat06,Kra07}.

In Section~\ref{Related_Literature} we provide a snapshot of related literature. In Section~\ref{Network_Model_sec} we introduce a channel model which captures the half-duplex constraint in a simple way. We introduce a capacity achieving coding strategy in Section~\ref{Section3}. The strategy is based on allocating the transmission and reception time slots of a node in dependence of the node's previously received data. The proposed strategy requires synchronization on the symbol level through a shared clock. In Section~\ref{Section4}, the performance of the coding strategy is analyzed yielding several capacity results. In the case of a relay cascade with a single source, it is shown that the coding strategy is capacity achieving, i.e. approaches a rate equal to
\begin{equation}
C_{m-1}(q)=\max_{p_{X_0 \dots X_m}}\min_{1\leq i \leq m}H(Y_i|X_i)
\label{Intro:Single_Source_C}
\end{equation}
where $m-1$ indicates the number of relays in the cascade and $X_i$ and $Y_i$ are the sent and received symbol of the $i$th relay. If the cascade includes a source and a certain number of relays with their own information, the strategy achieves the cut-set bound given that the rates of the relay sources fall below certain thresholds. Hence, a partial characterization of the boundary of the capacity region follows. For cascades composed of an infinite number of half-duplex constrained relays, we show that the capacity in bits/use (abbreviated as b/u in the remainder) is given by
\begin{equation}
C_\infty(q)=\log\left (\frac{1+\sqrt{4q+1}}{2}\right).
\end{equation}
Remarkably, $C_\infty(1)$ is equal to the logarithm of the golden ratio and $C_\infty(2)$ is $1$ b/u. In Section~\ref{Section5} the capacity results are applied to various special cases. In particular, we transform (\ref{Intro:Single_Source_C}) into a convex optimization program with linear objective and provide numerical solutions for $C_{m-1}(q)$ for different values of~$m$ and~$q$. Further, the single relay channel with a source and a relay source and binary transmission alphabet is considered and an explicit expression of the cut-set bound and of the achievable segment on the cut-set bound is computed. We finally show that the proposed coding strategy can be applied to wireless trees and to the half-duplex constrained butterfly network. In the latter case the proposed timing strategy outperforms the well-known XOR-based network coding strategy. 

\section{Related Literature}\label{Related_Literature}
The classical relay channel goes back to van der Meulen~\cite{meu71}. Further significant results concerning capacity and coding were obtained by Cover and El Gamal in~\cite{CovGam79}. A comprehensive literature survey as well as a classification of various \textit{decode-and-forward} and \textit{compress-and-forward} strategies for relay channels and small multiple relay networks is given in \cite{KraGasGup05}. General relay networks are very difficult to analyze (even the capacity of the non-degraded single relay channel is an open question). Motivated by the fact that line networks are often more accessible for analysis and, further, are fundamental building blocks of general communications systems, various source and channel coding problems have been examined without the assumption of half-duplex constrained nodes. 

Yamamoto~\cite{Yam81} considers a deterministic three node line network where the first node generates two random sequences. The region of achievable rates is found such that the second node is able to reconstruct the first sequence and the third node the second sequence within prescribed distortion tolerances. These results are extended to longer lines and branching communication systems in the same paper. A related version of the three node source coding problem is investigated in~\cite{VasTiaDig06}. The encoder at the first node intends to communicate a random sequence within certain distortion constraints to the relay and the destination under the assumption that the relay and the destination have access to individual side information about the source. The authors derive inner and outer bounds for the rate-distortion region and characterize scenarios where both bounds coincide. A \textit{distributed} source coding problem for the three node line network is examined in~\cite{CufSuGam09}. In contrast to the cases before, the relay acts as a source which is correlated to the source at the first node. The task of the destination is to estimate a function of the output of the two sources. Inner and outer bounds on the achievable rate region are provided such that an arbitrarily chosen distortion constraint is satisfied.

The channel capacity of three node line networks composed of two \textit{identical} binary channels where no processing is allowed at the middle terminal was examined in an early work~\cite{Sil55}. The author asks which channel of the infinite set of binary channels with equal capacity has to be cascaded with itself in order to achieve the largest end-to-end capacity. The answer is that a symmetric binary channel has a higher capacity under cascade than an asymmetric channel with the same capacity, unless the channels have very low capacity. Finite length cascades of \textit{identical} discrete memoryless channels are considered in~\cite{Sim70} under the assumption that the intermediate terminals do not possess any processing capability and that the transition matrix of the subchannels is nonsingular. By means of the eigenvalue decomposition of the transition matrix, the channel capacity is derived. Another work in which cascades composed of \textit{identical} discrete memoryless channels are investigated is \cite{NieFraTun07}. However, it is assumed that the intermediate relay nodes are able to process blocks of a fixed length. It is then shown that the capacity of the infinite length cascade equals the rate of the zero-error code of the underlying channel and that the capacity is always upper-bounded by the zero-error capacity of the underlying channel. In~\cite{KieCof93} the problem of finding the optimal ordering of a set of $n$ (\textit{distinct}) binary channels is analyzed such that the capacity of the resulting cascade is maximized. The question results from the observation that ordering has a strong influence on the capacity because matrix multiplication is not commutative. In the case of binary channels with positive determinants the authors are able to specify the optimal ordering. A line network composed of erasure channels is considered in \cite{PakFraSho05} for a single source-destination pair. The authors propose coding schemes which are based on fountain codes. 

In the work at hand we apply the idea of timing to half-duplex line networks. Timing is not a new idea in the information theoretical literature and has already been used in conjunction with queuing channels. Anantharam and Verd\'{u} showed~\cite{AnaVer96} that encoding information into the time differences of arrival to the queue achieves the capacity of the single server queue with exponential service distribution. The discrete-time version of this problem was analyzed in \cite{BedAzi98}. In \cite{Kra04}, Kramer developed a memoryless half-duplex relay channel model and computed decode and forward rates due to Cover and El~Gamal \cite{CovGam79}. He noticed that higher rates are possible when the transmission and reception time slots of the relay are random since one can send information through the timing of operating modes.

\section{Network Model and Information Flow} \label{Network_Model_sec}
\begin{figure*}[!t]
\psfrag{X0}{\small{\textrm{$X_0$}}}
\psfrag{Xi}{\small{\textrm{$X_i$}}}
\psfrag{Xi-1}{\small{\textrm{$X_{i-1}$}}}
\psfrag{Yi}{\small{\textrm{$Y_i$}}}
\psfrag{Yi-1}{\small{\textrm{$Y_{i-1}$}}}
\psfrag{X_m}{\small{\textrm{$X_{m}$}}}
\psfrag{Y_m}{\small{\textrm{$Y_{m}$}}}
\psfrag{Relay}{{\textrm{$\small{\textrm{Relay}}$}}}
\psfrag{i}{\small{\textrm{${i}$}}}
\psfrag{i-1}{\small{\textrm{${i-1}$}}}
\psfrag{1}{\scriptsize{$1$}}
\psfrag{2}{\scriptsize{$2$}}
\centering
\epsfig{file=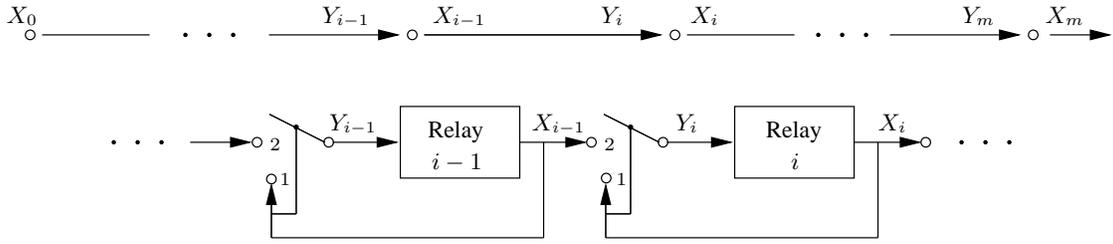, scale = 1}
\caption{A noiseless relay cascade and the link model illustrated by means of feedback. If relay~$i$ is transmitting, the switch is in position~$1$ otherwise in position~$2$. }
\label{fig:sys_model}
\end{figure*}
\subsection{Network Model}\label{Network_Model}
Consider the discrete memoryless relay cascade as depicted in Fig.~\ref{fig:sys_model}. The underlying topology corresponds to a directed path graph in which each node is labeled by a distinct number from $\mathcal{V}=\{0,\dots,m\}$ with $m>1$. The integers $0$ and $m$ belong to the first source and the sink, respectively, while all remaining integers $1$ to $m-1$ represent half-duplex constrained relays, i.e. relays which cannot transmit and receive at the same time. The connectivity within the network is described by the set of edges $\mathcal{E}=\{i \times (i+1):0\leq i \leq m-1\}$, i.e. the ordered pair $i \times (i+1)$ represents the communications link from node $i$ to node $i+1$. The output of the $i$th node, which is the input to channel $i \times (i+1)$ is denoted as $X_i$ and takes values on the alphabet $\mathcal{X}=\{0,\dots,q-1\} \cup \{\textrm{N}\}$ where $\mathcal{Q}=\{0,\dots,q-1\}$ denotes the $q$-ary transmission alphabet while ``N'' is meant to signify a channel use in which node~$i$ is not transmitting. The input of the $i$th node, which is the output of channel $(i-1) \times i$ is denoted as $Y_i$ and is given by
\begin{equation}
Y_{i}=\left\{\begin{array}{ll} X_{i-1}, & \mbox{if }X_{i}=\textrm{N}\\ X_{i}, & \mbox{if }X_{i} \in \mathcal{Q} \end{array}\right.
\label{relay_ch_model}
\end{equation}
where $1 \leq i \leq m$. Channel model (\ref{relay_ch_model}) captures the half-duplex constraint as follows. Assume relay $i$ is in transmission mode, i.e. $X_i \in \mathcal{Q}$. Then relay~$i$ hears itself ($Y_i=X_i$) but cannot listen to node~$i-1$ or, equivalently, relay $i$ and node~$i-1$ are disconnected. However, if relay $i$ is not transmitting, i.e. $X_i=\textrm{N}$, it is able to listen to relay~$i-1$ via a noise-free $(q+1)$-ary pipe ($Y_i=X_{i-1}$). The sink listens all the time, i.e. $X_m$ is always equal to \textrm{N}, and therefore its input is given by $Y_{m}=X_{m-1}$. Another interpretation of the channel model is that the output $X_i$ of relay~$i$ controls the position of a switch which is placed at its input. If relay~$i$ is transmitting, the switch is in position~$1$ otherwise it is in position~$2$ (see Fig.~\ref{fig:sys_model}). Since a pair of nodes is either perfectly connected or disconnected, we obtain a deterministic network with $p_{Y_1\dots Y_{m}|X_0 \dots X_m} \in\{0,1\}$ that factors as $\left[\prod_{i=1}^{m-1}p_{Y_i|X_iX_{i-1}}\right]p_{Y_m|X_{m-1}}$ where $p_{Y_i|X_iX_{i-1}}$ is defined by (\ref{relay_ch_model}).

\subsection{Information Flow}\label{Info_flow}
Every node $v \in \{0,\dots,m-1\}$ draws its messages uniformly and independently from the message set $\mathcal{W}_v=\left \{1,\dots,2^{nR_v}\right\}$ where $w_{v}(b)$ denotes the message sent by node~$v$ to node~$v+1$ in block~$b$. Each block has a length of~$n$. Observe that this setup includes the case that only a subset of the relays communicate own information to the sink by setting the rate of the remaining relays to zero. The relays allocate information to the codewords as follows. At the end of block~$b-1$, each relay $v$ with a rate $R_v>0$ carries out two tasks. It draws a new message $w_{v}(b)$ and it decodes the messages $\{w_{0}(b-v),\dots,w_{v-1}(b-1)\}$ from the received codeword $\mathbf{y}_{v}(b-1)$. The new message together with the decoded messages are forwarded to node~$v+1$ in block~$b$ by means of the sequence $\mathbf{x}_v(b)$. Similarly, each relay $v$ without own information, i.e. $R_v=0$, decodes the messages $\{w_{0}(b-v),\dots,w_{v-1}(b-1)\}$ at the end of block~$b-1$ and forwards the decoded messages to the next node~$v+1$ by means of $\mathbf{x}_v(b)$. Source node~$0$ sends one message $w_0(b)$ per block represented through $\mathbf{x}_0(b)$. We assume an initialization period of $m-1$ blocks. In the first block node~$0$ forwards information, in the second block nodes~$0$ and $1$ forward information and so forth. From the $m$th block onwards all nodes (except of the sink) forward information. Thus, the sink does not decode until the end of the $m$th block. Since a very large number of transmission blocks is considered, it is allowed to neglect the initial delay in an asymptotic analysis. In the next paragraph, a coding strategy is introduced which realizes the outlined information flow.

\section{A Timing Code for Line Networks with Multiple Sources} \label{Section3}
\subsection{General Idea and Codebook Sizes} \label{sec:coding_idea}
A coding strategy is introduced which relies on the observation that information can be represented not only by the value of code symbols but also by the position of code symbols, i.e. by timing the transmission and reception slots of the relay nodes. The strategy requires synchronization on the symbol level through a shared clock. The codebook construction is recursive and guarantees that adjacent nodes do not transmit at the same time. The following encoding techniques are applied at the source and the relays where $n_i$ denotes the number of transmitted symbols of node~$i$ within one block of~$n$ symbols.
\begin{itemize}
\item At relay $m-1$: Relay $m-1$ represents information by choosing $n_{m-1}$ transmission symbols per block from the $q$-ary transmission alphabet $\mathcal{Q}$ combined with allocating the $n_{m-1}$ symbols to the transmission block of $n$ symbols. Thus, $q^{n_{m-1}} {n\choose n_{m-1}}$ different sequences $\mathbf{x}_{m-1}$ of length $n$ are available at relay $m-1$. Observe that $q^{n_{m-1}}$ equals the number of possible distinct sequences when the $q$-ary symbols are located at fixed slots while ${n \choose n_{m-1}}$ equals the number of possible transmission-listen patterns. 
\item At relay $i$, $1\leq i \leq m-2$: Observe that the effective codeword length of relay~$i$ reduces to $n-n_{i+1}$ since relay~$i+1$ cannot listen to relay~$i$ when it (relay~$i+1$) transmits. For each transmission-listen pattern used by node $i+1$, node $i$ generates $q^{n_i} {n-n_{i+1}\choose n_{i}}$ different sequences by allocating $n_i$ transmission symbols from the alphabet $\mathcal{Q}$ in all possible ways to the $n-n_{i+1}$ listen slots of the pattern. The remaining slots of the pattern, i.e. the slots in which node $i+1$ transmits, are filled with idle symbols ``N''. As before, $q^{n_i}$ equals the number of possible distinct sequences when the $q$-ary symbols are located at fixed slots while ${n-n_{i+1}\choose n_{i}}$ equals the number of possible transmission-listen patterns. The procedure generates a certain number of transmission-listen patterns used by node~$i$. 
\item At source node $0$: The source uses the $(q+1)$-ary alphabet~$\mathcal{X}=\mathcal{Q}\cup \{\textrm{N}\}$ for encoding without transmitting information in the timing of the symbols. Hence, the non-transmission symbol ``N'' is used as a regular alphabet symbol. Due to the half-duplex constraint at relay~$1$, the effective codeword length of the source reduces to $n-n_{1}$ what results from the fact that relay $1$ cannot pay attention to the source when it (relay $1$) transmits. Thus, the source is able to generate $(q+1)^{n-n_1}$ different sequences $\mathbf{x}_0$.
\end{itemize}

Next, the maximum size of $\mathcal{W}_0$, $\mathcal{W}_1,\dots,$ $\mathcal{W}_{m-1}$ is given. From the previous paragraph, we immediately obtain
\begin{equation}
|\mathcal{W}_0| \leq (q+1)^{n-n_1}.
\label{|W_0|}
\end{equation}
Both the source and the relays choose their messages uniformly and independently of each other. Hence, relay~$v$ is required to reserve $\prod_{i=0}^{v-1}|\mathcal{W}_{i}|$ sequences in order to represent an arbitrary combination of arriving messages $\{w_{0}(b-v),\dots,w_{v-1}(b-1)\}$. Arriving messages are encoded by each relay~$v$ with transmission patterns and a fixed number $k_v\in \{0,\dots,n_v\}$ of transmission symbols. To be more precise, each combination of arriving messages is assigned injectively to a subset of the set of all the sequences $\mathbf{x}_{v}$. The subset comprises those sequences $\mathbf{x}_{v}$ such that all transmission patterns occur and such that the first $k_v$ transmission symbols of each transmission pattern take all possible values. The remaining $n_v-k_v$ transmission symbols per transmission pattern are used by relay~$v$ for encoding own messages $w_{v}(b)$. With the foregoing explanation in mind, we have for all $v \in \mathcal{V}\setminus{\{0,m\}}$
\begin{equation}
\prod_{i=0}^{v-1}|\mathcal{W}_i| \leq q^{k_v} {n -n_{v+1}\choose n_{v}}
\label{|W_relay_source|}
\end{equation}
and 
\begin{equation}
|\mathcal{W}_{v}|\leq q^{n_v-k_v}.
\label{|W|-constraint}
\end{equation}
If relay~$v$ does not have own information, then $k_v=n_v$. As a final remark, transmission patterns can only be used for encoding arriving messages. Otherwise, if relay~$v$ would encode own messages $w_{v}(b)$ by means of transmission patterns, node~$v-1$ would not know when node~$v$ listens in block~$b$ as $w_v(b)$ (hence the transmission pattern used by node~$v$) is not known by node~$v-1$.
\subsection{Example}
We now illustrate the ideas introduced in the previous section by constructing a code for a relay cascade with four nodes, i.e. $\mathcal{V}=\{0,\dots,3\}$, where nodes~$0$ and $2$ act as sources with a rate greater than zero. The transmission alphabet is binary, i.e. $q=2$, and the code parameters are $n=4$, $n_1=1$, $n_2=2$ (and $n_3=0$ of course). According to (\ref{|W_0|}) to (\ref{|W|-constraint}), the maximum size of the message sets is $|\mathcal{W}_0|=|\mathcal{W}_2|=4$ obtained for $k_1=1$ and $k_2=0$, which corresponds to a sum rate of $1$~b/u. Table~\ref{Tab:Ex_Code} depicts possible codebooks $\mathcal{C}_{0}$, $\mathcal{C}_{1}$, $\mathcal{C}_{2}$ for nodes $0$, $1$ and $2$, respectively, and Table~\ref{cod_example_2} shows how to use the codebooks in order to send a particular message sequence. 

Let us first consider $\mathcal{C}_{2}$ which consists of $16$ different codewords. The four underlying transmission patterns (arbitrarily chosen from the ${4\choose 2}$ possible patterns) are shown in the last column of Table~\ref{Tab:Ex_Code}. Each transmission pattern is identified with a unique color $r \in \{a,b,c,d\}$ and the $n_2=2$ binary transmission slots within each pattern are marked with $\textrm{B,C}\in\{0,1\}$. Node~$2$ uses the transmission patterns for representing source messages $w_0$. In detail, pattern~$a$ represents $w_0=0$, pattern~$b$ represents $w_0=1$ and so forth. Own messages $w_2$ are encoded by the transmission symbols \textrm{B} and \textrm{C} according to $w_2 \mapsto (\textrm{B},\textrm{C})$:  $0 \mapsto (0,0)$, $1 \mapsto (0,1)$, $2 \mapsto (1,0)$, $3 \mapsto (1,1)$. 
\begin{table*}[!t]
\renewcommand{\arraystretch}{1.3}
\begin{center}
\caption{}
\subtable[Example codebooks for source, relay and relay source]{
\begin{tabular}{c||c c c|c c c c| c} 
$w_0$& \multicolumn{3}{c|}{$\mathcal{C}_{0}$}&\multicolumn{4}{c|}{$\mathcal{C}_{1}$}&$\mathcal{C}_{2}$ \\ \hline
$0$&$\textrm{N}0\textrm{N}\textrm{N}$ $e$ &$0\textrm{N}\textrm{N}\textrm{N}$ $f$& $0\textrm{N}\textrm{N}\textrm{N}$ $g$&$0\textrm{N}\textrm{N}\textrm{N}$ $(a,e)$&$\textrm{N}0\textrm{N}\textrm{N}$ $(b,f)$&$0\textrm{N}\textrm{N}\textrm{N}$  $(c,e)$&$\textrm{N}0\textrm{N}\textrm{N}$ $(d,f)$&NBNC $a$ \\
$1$&$\textrm{N}1\textrm{N}\textrm{N}$ $e$&$1\textrm{N}\textrm{N}\textrm{N}$ $f$& $1\textrm{N}\textrm{N}\textrm{N}$ $g$&$1\textrm{N}\textrm{N}\textrm{N}$ $(a,e)$&$\textrm{N}1\textrm{N}\textrm{N}$ $(b,f)$&$1\textrm{N}\textrm{N}\textrm{N}$ $(c,e)$&$\textrm{N}1\textrm{N}\textrm{N}$ $(d,f)$& BNCN $b$\\
$2$&$\textrm{N}\textrm{N}0\textrm{N}$ $e$&$\textrm{N}\textrm{N}0\textrm{N}$ $f$& $\textrm{N}0\textrm{N}\textrm{N}$ $g$&$\textrm{N}\textrm{N}0\textrm{N}$ $(a,g)$&$\textrm{N}\textrm{N}\textrm{N}0$ $(b,g)$&$\textrm{N}\textrm{N}\textrm{N}0$ $(c,g)$&$\textrm{N}\textrm{N}0\textrm{N}$ $(d,g)$& NBCN $c$\\
$3$&$\textrm{N}\textrm{N}1\textrm{N}$ $e$&$\textrm{N}\textrm{N}1\textrm{N}$ $f$& $\textrm{N}1\textrm{N}\textrm{N}$ $g$&$\textrm{N}\textrm{N}1\textrm{N}$ $(a,g)$&$\textrm{N}\textrm{N}\textrm{N}1$ $(b,g)$&$\textrm{N}\textrm{N}\textrm{N}1$ $(c,g)$&$\textrm{N}\textrm{N}1\textrm{N}$ $(d,g)$& BNNC $d$\\ \hline
& \multicolumn{3}{c|}{}&\multicolumn{4}{c|}{}& \\ 
\end{tabular}
\label{Tab:Ex_Code}
}
\subtable[Illustration how to use the code]{
\begin{tabular}{c|| c| c| c| c| c| c| c} \hline
block $i$&$w_0(i)$&$w_2(i)$&  $\mathbf{x}_0(i)$ &$\mathbf{x}_1(i)$& $\mathbf{x}_2(i)$ & $\hat{w}_0(i)$ & $\hat{w}_2(i)$\\ \hline\hline
$\parbox[0pt][1em][c]{0cm}{}i=1$&$3$&-& $\textrm{N}\textrm{N}1\textrm{N}$ &$\textrm{NNNN}$& $\textrm{NNNN}$ & - & -\\ 
$\parbox[0pt][1em][c]{0cm}{}i=2$&$1$&-& $1\textrm{N}\textrm{N}\textrm{N}$ &$\textrm{N}\textrm{N}1\textrm{N}$& $\textrm{NNNN}$ & - & - \\ 
$\parbox[0pt][1em][c]{0cm}{}i=3$&$2$&$0$& $\textrm{N}\textrm{N}0\textrm{N}$ &$\textrm{N}1\textrm{N}\textrm{N}$&$0\textrm{N}\textrm{N}0$ & $3$ &$0$\\ 
$\parbox[0pt][1em][c]{0cm}{}i=4$&-&$2$& $\textrm{NNNN}$ &$\textrm{N}\textrm{N}\textrm{N}0$& $1\textrm{N}0\textrm{N}$ & $1$ &$2$\\ 
$\parbox[0pt][1em][c]{0cm}{}i=5$&-&$3$& $\textrm{NNNN}$ &$\textrm{NNNN}$& $\textrm{N}11\textrm{N}$ & $2$ &$3$\\ \hline
\end{tabular}
\label{cod_example_2}
}
\end{center}
\end{table*}

Next, $\mathcal{C}_{1}$ is considered. Recall that $\mathcal{C}_{1}$ has to be constructed such that node~$1$ is able to represent one out of four possible source node~$0$ messages per block independently from the transmission pattern used by node~$2$ in the same block. Hence, four codewords per transmission pattern~$a$, $b$, $c$ and $d$ have to be constructed. Take, for instance, pattern~$a$. When node~$2$ uses pattern~$a$, node~$1$ can encode its information in slots one and three. The following mapping $w_0 \mapsto (x_{1,1}, x_{1,3})$ is chosen for encoding where $x_{1,1}, x_{1,3} \in \{0,1,\textrm{N}\}$ indicate the symbols used by node~$1$ in slots one and three: $0 \mapsto (0,\textrm{N})$, $1 \mapsto (1,\textrm{N})$, $2 \mapsto (\textrm{N},0)$, $3 \mapsto (\textrm{N},1)$. Note that this mapping includes timing. By allocating each of the four values of $(x_{1,1}, x_{1,3})$ to the listen slots of pattern $a$ and, further, by requiring that node $1$ is quiet when node $2$ transmits (i.e. allocating ``\textrm{N}'' to slots~$2$ and $4$), we obtain the codewords in the first column of $\mathcal{C}_{1}$. Applying the same procedure to pattern~$b$, $c$ and $d$ yields column two, three and four of $\mathcal{C}_1$. The label $(r,s)\in \{a,b,c,d\}\times\{e,f,g\}$ next to each codeword in~ $\mathcal{C}_{1}$ has the following meaning. The first color indicates the transmission pattern in~$\mathcal{C}_{2}$ from which the codeword was constructed while the second color indicates the transmission pattern of the codeword in~$\mathcal{C}_{1}$. 

Finally, we consider $\mathcal{C}_0$. In each transmission block, source node $0$ can use three time slots $t_1$, $t_2$ and $t_3$ for encoding since node $1$ sends once per block. Let $x_{0,t_1}, x_{0,t_2}, x_{0,t_3} \in \{0,1,\textrm{N}\}$ denote the symbols used by node~$0$ for encoding a particular message $w_0 \in \mathcal{W}_0$. We use the mapping $w_0 \mapsto (x_{0,t_1}, x_{0,t_2}, x_{0,t_3})$ for encoding where $0 \mapsto (0,\textrm{N},\textrm{N})$, $1 \mapsto (1,\textrm{N},\textrm{N})$, $2 \mapsto (\textrm{N},0,\textrm{N})$, $3 \mapsto (\textrm{N},1,\textrm{N})$. Again, the mapping includes timing. Now, by allocating all possible values of $(x_{0,t_1}, x_{0,t_2},x_{0,t_3})$ to the listen slots of codewords in $\mathcal{C}_1$ whose second color is $s\in \{e,f,g\}$ and, further, by requiring that node~$0$ is quiet when node $1$ transmits, we obtain all codewords in $\mathcal{C}_{0}$ which are colored with $s$. It should be noted that merely four from $27$ possible sequences are used in the mapping $w_0 \mapsto (x_{0,t_1}, x_{0,t_2}, x_{0,t_3})$. Hence, $\mathcal{C}_0$ could be designed such that node~$0$ is able to send $\lfloor27/4\rfloor$ additional messages to a sink at node~$1$ at a rate of $0.6462$~b/u. 

Observe that adjacent nodes are able to cooperate since each node knows the message(s) to be forwarded by the next node as well as the coding strategy applied by the next node. Hence, a node is always aware of the codeword used by the next node and, therefore, can pick a codeword from the correct column of its codebook. In particular, the codewords for block~$i$ are picked as follows. The encoder at node~$0$ determines, based on message $w_0(i-2)$, the color of $\mathbf{x}_2(i)$ and, therefore, knows the first color~$r$ of codeword $\mathbf{x}_1(i)$. Then, based on this information, the encoder at node~$0$ determines the second color $s$ of $\mathbf{x}_1(i)$ by means of $w_0(i-1)$. This color tells node~$0$ from which column in $\mathcal{C}_0$ $\mathbf{x}_0(i)$ has to be picked, namely from a column whose codewords are colored with $s$. The precise choice within the picked column depends on the new source message $w_0(i)$. Similarly, the encoder at node~$1$ determines, based on message $w_0(i-2)$, color~$r$ of $\mathbf{x}_2(i)$ and, therefore, knows that $\mathbf{x}_1(i)$ has to be picked from a column of $\mathcal{C}_1$ whose entries have~$r$ as their first color. The precise choice within the column depends on message $w_0(i-1)$. The encoder at node~$2$ knows $\left \{w_0(i-2),w_2(i)\right\}$ at the beginning of block~$i$. Message $w_0(i-2)$ tells him which transmission pattern to use in $\mathbf{x}_2(i)$ while $w_2(i)$ determines the transmission symbols. 

We conclude the example by demonstrating how the codebooks $\mathcal{C}_0$, $\mathcal{C}_1$ and $\mathcal{C}_2$ have to be used such that source node~$0$ is able to transmit messages $3,1,2$ to the sink while relay source~$2$ transmits messages $0,2,3$ to the sink. Note that the transmission strategy includes the arrangement that a node picks its very first codeword from the first column of its codebook. The result is shown in Table~\ref{cod_example_2}. 

\subsection{Rate Region}
\label{rate_region}
We now determine an achievable rate region $\mathcal{R}$ from the expressions derived in section \ref{sec:coding_idea}. All logarithms which will be used in the following are to base~$2$. As usual, $R_{i} = \log|\mathcal{W}_{i}|/n$. In order to avoid tedious case distinctions we assume $R_0>0$ in the remainder. Hence, $0< n_i <n$ for all $1 \leq i \leq m-1$. This is without loss of generality since the rate region of a cascade with $R_0=0$ is equal to the rate region of the shortened cascade where the first node with a rate greater zero is made to node~$0$. The following abbreviations are used for the portion of time in which relay~$i$ listens or transmits: $p_i\equiv n^{-1}(n-n_i)$ and $\bar{p}_i\equiv 1-p_i$. Observe that for $1 \leq i \leq m-1$
{\setlength{\arraycolsep}{0.0em}
\begin{eqnarray}
\label{constraint_on_pi}
0 < p_i &{}<{}& 1\\
\label{constraint_on_pi_pi+1}
p_i+p_{i+1}&{}\geq{}& 1\\
\label{constraint_on_pm}
p_m &{}={}&1
\end{eqnarray}} 
since $n_i\leq n-n_{i+1}$ and $n_m=0$ due to the code construction. The set of points characterized by (\ref{constraint_on_pi}) to (\ref{constraint_on_pm}) will be denoted as $\mathcal{P}^\star\subset \mathbb{R}^m$. By identifying $p_i$ with $p_{X_i}(N)$, we can regard $\mathcal{P}^\star$ as a subset of the joint probability distributions $p_{X_0 \dots X_m}$. Obviously, all distributions in $\mathcal{P}^\star$ factorize as $p_{X_0}p_{X_1|X_0}\dots p_{X_m|X_{m-1}}$.

The method of types~\cite{Csi98} provides important tools for relating combinatorial expressions to information theoretic expressions. An example which will be useful for the problem considered here is~\cite[Th. 1.4.5]{Lin99}
\begin{equation}
n^{-1}\log{n\choose n_{i}} = H\left (p_i\right) + o(1) \quad \textrm{for }n\rightarrow \infty
\label{binom_as_entropy}
\end{equation} 
where $H(p_i)$ denotes the binary entropy function evaluated at $p_i=n^{-1}(n-n_i)$. Using (\ref{binom_as_entropy}), we obtain from (\ref{|W_0|}) to (\ref{|W|-constraint}) for $n \rightarrow \infty$
{\setlength{\arraycolsep}{0.0em}
\begin{eqnarray}
\label{R_0}
R_0 &{}\leq{}& p_1\log(q+1) \\
\label{R_i}
\sum_{i=0}^{v} R_i &{}\leq{}& \bar{p}_v\log q+p_{v+1}H\left (\bar{p}_v p_{v+1}^{-1}\right) +o(1)\\
\label{coding_constraint}
R_{v} &{}\leq{}& \bar{p}_v\log q
\end{eqnarray}}
where $ v \in \mathcal{V}\setminus{\{0,m\}}$. As an aside, (\ref{R_i}) results from adding the logarithm of (\ref{|W_relay_source|}) to the logarithm of (\ref{|W|-constraint}), dividing the result by $n$ and applying (\ref{binom_as_entropy}). Inequality (\ref{R_i}) is well-defined since $p_{v+1}\neq 0$ and $\bar{p}_v p_{v+1}^{-1}\in (0,1]$ due to (\ref{constraint_on_pi}) to (\ref{constraint_on_pm}). 

The achievable rate region for $n \rightarrow \infty$ is given by
\begin{equation}
\mathcal{R}=\mathit{Co}\left ( \bigcup_{p \in \mathcal{P}^\star}\mathcal{R}_p\right)
\label{Rate_region_rough}
\end{equation}
where $\mathcal{R}_p$ indicates the region resulting from (\ref{R_0}) to (\ref{coding_constraint}) for a particular point $p \in \mathcal{P}^\star$ while the convex hull $\mathit{Co}(\cdot)$ takes time-sharing between different regions $\mathcal{R}_p$ into account.

Conditions (\ref{R_0}) to (\ref{coding_constraint}) are merely another formulation of conditions (\ref{|W_0|}) to (\ref{|W|-constraint}) for $n\rightarrow \infty$. Since we can construct codebooks of the size stated in (\ref{|W_0|}) to (\ref{|W|-constraint}) by means of the outlined procedure, it immediately follows that the rates due to (\ref{R_0}) to (\ref{coding_constraint}) are achievable and, thus, the conditions are sufficient.

\section{Capacity Results} \label{Section4}
In this section we shall investigate the optimality of the coding strategy. We will make use of the following notation. The complement of a set $S$ within an ambient set is denoted as $S^c$, the power set of a set $S$ is denoted as $\mathcal{P}(S)$ and $X_S:=\{X_i:i \in S\}$ indicates a set of random variables. Further, $\mathbf{R}$ is a $|\mathcal{V}|-1$-dimensional rate vector with $R_v$ as its $v$th entry. We will use pmf as acronym for probability mass function.

A well-known result, which bounds the rate of information flow from nodes in $S^c$ to nodes in $S$ is the so-called \textit{cut-set bound}.
\begin{lem}[Cut-Set Bound]{\cite[chap.~14.10]{CovTho91}}
Consider a general multiterminal network composed of $m+1$ nodes and channel $p_{Y_0 \dots Y_m|X_0 \dots X_m}$. $R_{ij}$ denotes the transmission rate between two nodes $i$ and $j$. If the information rate $\left(R_{ij}\right)$ is achievable, then there is some joint probability distribution $p_{X_0 \dots X_{m}}$, such that
\begin{equation}
\sum_{i\in S^c, j \in S} R_{ij}\leq I\left(X_{S^c};Y_{S}|X_{S} \right),
\end{equation}
for all $S \subset \{0,\dots,m\}$. 
\label{Lemma1}
\end{lem}
\begin{lem}
Consider a noise-free relay cascade as described in section~\ref{Network_Model}. If the information rate $\left(R_v\right)$ is achievable, then there is some joint probability distribution $p_{X_v\dots X_m}$, such that 
\begin{equation}
\sum_{k=0}^{v}R_k \leq  \max_{p_{X_v\dots X_m}}\min_{v+1\leq i\leq m} H(Y_i|X_i)
\end{equation}
for all $v\in \mathcal{V}\setminus \{m\}$.
\label{Lemma2}
\end{lem}
\begin{proof}
We determine a sufficient subset from the set of all possible network cuts. An upper bound on the sum rate $\sum_{k=0}^{v}R_{k}$ due to Lemma~\ref{Lemma1} is given by
\begin{equation}
\sum_{k=0}^{v}R_{k} \leq \max_{p_{X_v\dots X_m}} \min_{S \in \mathcal{M}} I(X_v,X_{S^c};Y_{S},Y_{m}|X_{S}),
\label{cutset_bound_single_SD}
\end{equation}
where $\mathcal{M} = \mathcal{P}\left(\{v+1,\dots,m-1\}\right)$ and $S^c$ is the complement of $S$ in $\{v+1,\dots,m-1\}$. We further have
\begin{equation}
I(X_v,X_{S^c};Y_{S},Y_{m}|X_{S}) = H(Y_{S},Y_{m}|X_{S})
\label{cut-set_bound_deterministic}
\end{equation}
since the network is deterministic. Now suppose that $S$ is nonempty and let $i \in \{v+1,\dots,m-1\}$ denote the smallest integer in $S$. By the chain rule for entropy, we can expand $H(Y_{S},Y_{m}|X_{S})$ as
{\setlength{\arraycolsep}{0.0em}
\begin{eqnarray}
\label{cut_set_expansion}
H(Y_{S},Y_{m}|X_S) &{}={}& H(Y_i|X_S) + H(Y_{S\setminus\{i\}}|X_S,Y_i) \nonumber\\ 
&&{+}\:H(Y_{m}|X_S,Y_S)\\ 
\label{cut_set_expansion_lb}
&{}\geq{}& H(Y_i|X_S). \nonumber
\end{eqnarray}}
For each cut $S$ with smallest entry~$i$, a cut called $S_i$ can be found such that $H(Y_{S_i},Y_{m}|X_{S_i})$ is less than or equal to $H(Y_{S},Y_{m}|X_{S})$. Simply choose $S_i:=\{i,\dots,m-1\}$. This eliminates the second and third term on the right hand side of (\ref{cut_set_expansion}) due to the underlying channel model (\ref{relay_ch_model}). Further, since $S \subseteq S_i$ we have $H(Y_i|X_S) \geq H(Y_i|X_{S_i})$. Thus, each non-empty cut $S$ with smallest element $i$ is dominated by $S_i$ in terms of delivering a smaller entropy value. Finally, $S =\emptyset$ has to be considered in (\ref{cut-set_bound_deterministic}) which yields $H(Y_{m})$. To sum up, $\sum_{k=0}^{v}R_{k}$ is upper bounded by\footnote{Note that $H(Y_m)=H(Y_m|X_m)$. For notational convenience, we will always use $H(Y_m|X_m)$.}
{\setlength{\arraycolsep}{0.0em}
\begin{eqnarray}
\label{C_upper_bound_1steq}
\sum_{k=0}^{v}R_{k} &{}\leq{}& \max_{p_{X_v\dots X_m}} \min_{v+1 \leq i \leq m} H(Y_i|X_{S_i}) \\
&{}\leq{}& \max_{p_{X_v\dots X_m}} \min_{v+1 \leq i \leq m} H(Y_i|X_i)
\label{C_upper_bound}
\end{eqnarray}}
where the last inequality follows from the fact that conditioning does not increase entropy.
\end{proof}
\begin{thm}
The capacity of a noise-free relay cascade with a single source-destination pair (namely nodes $0$ and $m$) and $m-1$ half-duplex constrained relays is given by
\begin{equation}
C_{m-1}(q)=\max_{p_{X_0 \dots X_m}}\min_{1\leq i \leq m}H(Y_i|X_i)
\label{eq:Corollary1}
\end{equation}
where the maximization is over all $p_{X_0 \dots X_m}$  as shown in Table~\ref{p_Xv-1_Xv} and \ref{p_X0_X1} and $q$ equals the number of transmission symbols. Under consideration of the optimal input distribution stated in Table~\ref{p_Xv-1_Xv} and \ref{p_X0_X1}, (\ref{eq:Corollary1}) becomes (\ref{eq:th1_explicit})
\begin{figure*}[!b]
\vspace*{4pt}
\hrulefill
\begin{equation}
\label{eq:th1_explicit}
C_{m-1}(q)=\max_{p_1,\dots,p_{m-1}}\min \left\{p_1\log(q+1),\min_{1\leq i \leq m-1} \left\{\bar{p}_{i}\log q+p_{i+1}H\left(\bar{p}_i p_{i+1}^{-1}\right)\right\}\right\} 
\end{equation}
\end{figure*}
where $0 < p_i < 1$, $p_{m}=1$ and $p_i + p_{i+1} \geq 1$ for all $i \in \{1,\dots,m-1\}$.
\label{Theorem1}
\end{thm}
\begin{proof}
By Lemma~\ref{Lemma2} we have
\begin{equation}
C_{m-1}(q)\leq\max_{p_{X_0 \dots X_m}}\min_{1\leq i \leq m}H(Y_i|X_i).
\label{upper_bound_proof_th1}
\end{equation}
The opposite direction of (\ref{upper_bound_proof_th1}) is shown as follows. Consider the marginal pmf $p_{X_0X_1},\dots,$ $p_{X_{m-1}X_m}$ given in Table~\ref{p_Xv-1_Xv} and \ref{p_X0_X1}. We show that these functions are optimal in terms of maximizing $H(Y_i|X_i)$, $i \geq 1$.
\begin{table}[th]
\renewcommand{\arraystretch}{1.3}
\begin{center}
\caption{}
\subtable[Optimal $p_{X_{i-1}X_i}$ for $2 \leq i \leq m$]{
\begin{tabular}{c|c c c c}
\multicolumn{1}{c}{\backslashbox{\normalsize{$X_{i-1}$}}{\normalsize{$X_i$}}}&$0$ &$\cdots$ &$q-1$ & N \\ \cline{2-5}
$0$& $0$ & $\cdots$ & $0$ & $\bar{p}_{i-1}/q$\\ 
$\vdots$& $\vdots$&$\ddots$ &$\vdots$ &$\vdots$\\ 
$q-1$& $0$ & $\cdots$ & $0$ &$\bar{p}_{i-1}/q$\\ 
N&$\bar{p}_{i}/q$ & $\cdots$ &$\bar{p}_{i}/q$ &$p_i-\bar{p}_{i-1}$ \\ 
\end{tabular}
\label{p_Xv-1_Xv}}
\subtable[Optimal $p_{X_0X_1}$]{
\begin{tabular}{c|c c c c}
\multicolumn{1}{c}{\backslashbox{\normalsize{$X_{0}$}}{\normalsize{$X_1$}}} &$0$ &$\cdots$ &$q-1$ & N \\ \cline{2-5}
$0$& $0$ & $\cdots$ & $0$ & $p_1/(q+1)$\\ 
 $\vdots$& $\vdots$&$\ddots$ &$\vdots$ &\multirow{2}{*}{$\vdots$}\\ 
 $q-1$& $0$ & $\cdots$ & $0$ &\\ 
 N&$\bar{p}_1/q$ & $\cdots$ &$\bar{p}_1/q$ & $p_1/(q+1)$\\ 
\end{tabular}
\label{p_X0_X1}}
\end{center}
\end{table}
The zero probabilities in Table~\ref{p_Xv-1_Xv} and \ref{p_X0_X1} result from the following well-known fact~\cite[Def.~3]{VanMeu92}: a channel input can be neglected if it produces the same channel output as another channel input and this with the same probabilities. Consider e.g. the first column in Table~\ref{p_Xv-1_Xv}. For all $k\in \mathcal{X}$, the inputs $(X_{i-1},X_i)=(k,0)$ produce $Y_i=0$ with probability $1$. Hence, all but one input can be neglected. Applying the same consideration to the second till $q$th column yields that only one non-zero entry remains in each of the first $q$ columns of Table~\ref{p_Xv-1_Xv} and \ref{p_X0_X1}. Let us now address the last column of Table~\ref{p_Xv-1_Xv}. Recall that a permutation of the transmission symbols $x_{i-1} \in \mathcal{Q}$ still yields the same information flow between two nodes~$i-1$ and $i$. Hence, $p_{X_{i-1}X_i}(k,\textrm{N})=p_{X_{i-1}X_i}(l,\textrm{N})$ can be chosen for all $k,l \in \mathcal{Q}$. Considering the relative frequency $\bar{p}_{i-1}$ of transmission symbols used by node~$i-1$, we have $p_{X_{i-1}X_i}(k,\textrm{N})=\bar{p}_{i-1}/q$ for all $k \in \mathcal{Q}$ where $2 \leq i \leq m$. 

In order to achieve the maximum information flow from source node~$0$ to relay~$1$, the source has to encode with uniformly distributed input symbols when relay~$1$ listens, i.e. $p_{X_0X_1}(k,\textrm{N})=p_{X_0X_1}(l,\textrm{N})$ for all $k,l \in \mathcal{X}$. By taking this additional constraint into account, we obtain the last column of Table~\ref{p_X0_X1}. 

The constraints on $p_i$, which are stated in the last line of the Theorem, are necessary in order to guarantee that Table~\ref{p_Xv-1_Xv} and \ref{p_X0_X1} are proper probability mass functions. It is now fairly easy to check that the following equalities hold
{\setlength{\arraycolsep}{0.0em}
\begin{eqnarray}
\label{H(X_0|X_1)}
H(Y_1|X_1)&{}={}&p_1\log(q+1) \\
\label{H(X_v|X_{v+1})}
H(Y_{i+1}|X_{i+1})&{}={}&\bar{p}_{i}\log q+p_{i+1}H\left(\bar{p}_i p_{i+1}^{-1}\right)
\end{eqnarray}}
for all $ 1\leq i \leq m-1$. Observe that the set of probability mass functions defined by Table~\ref{p_Xv-1_Xv} and \ref{p_X0_X1} is equal to $\mathcal{P}^\star$, i.e. the set of empirical distributions due to the code construction defined by (\ref{constraint_on_pi}) to (\ref{constraint_on_pm}). Further, by assumption, $R_i=0$ for all $i\in \mathcal{V}\setminus\{0\}$. Then, a comparison of (\ref{H(X_0|X_1)}) and (\ref{H(X_v|X_{v+1})}) with (\ref{R_0}) and (\ref{R_i}) reveals that $\min_{1\leq i \leq m}H(Y_i|X_i)$ is an achievable rate. Hence, the capacity is lower bounded by 
\begin{equation}
C_{m-1}(q)\geq\max_{p_{X_0 \dots X_m}}\min_{1\leq i \leq m}H(Y_i|X_i)
\label{lower_bound_proof_th1}
\end{equation}
where the maximization is with respect to Table~\ref{p_Xv-1_Xv} and \ref{p_X0_X1}. Inequality (\ref{lower_bound_proof_th1}) together with (\ref{upper_bound_proof_th1}) proves (\ref{eq:Corollary1}). Replacing the conditional entropies in (\ref{eq:Corollary1}) by (\ref{H(X_0|X_1)}) and (\ref{H(X_v|X_{v+1})}) gives (\ref{eq:th1_explicit}). 
\end{proof}
\begin{rem}
\begin{remenum}
\item A more intuitive explanation of the zero probability assignment in Table~\ref{p_Xv-1_Xv} and \ref{p_X0_X1} is the following. Assume relay $i$ is transmitting, i.e. $X_i \in \mathcal{Q}$. According to the underlying channel model, relay $i$ is not able to listen to the input of node~$i-1$ and, consequently, node $i-1$ should not transmit when node~$i$ transmits.
\item One could ask why the channel inputs $(X_{i-1},X_i)=(k,\textrm{N})$, $k \in \mathcal{Q}$ and $(X_{i-1},X_i)=(\textrm{N},\textrm{N})$ have equal probability mass for $i=1$ but not necessarily for $i>1$ since for $i>1$ the information flow between relay~$i-1$ and $i$ should also be maximized. However, in contrast to the source node, relay~$i-1$ receives information. The amount of received information depends on the fraction of listening time provided by relay~$i-1$. Thus, choosing uniformly distributed inputs  $(X_{i-1},X_i)=(k,\textrm{N})$, $k \in \mathcal{X}$, maximizes the rate on link $(i-1)\times i$ but eventually reduces the rate on link $(i-2)\times(i-1)$.
\item Capacity expression (\ref{eq:Corollary1}) in Theorem~\ref{Theorem1} could also have been obtained by applying the decode-forward rate of Xie and Kumar~\cite{XieKum05} to the model considered in this paper. However, we show achievability by a constructive argument while Xie and Kumar use a random coding argument in their proof.
\end{remenum}
\end{rem}

The capacity of a single source line network with an infinite number of half-duplex constrained relays is stated in Theorem~\ref{Theorem2}.
\begin{thm}
For $m\rightarrow \infty$, i.e. for an unbounded number of relays, and $q$ transmission symbols, the capacity of the noise-free and half-duplex constrained relay cascade with a single source-destination pair is equal to
\begin{equation}
C_\infty(q)=\log\left (\frac{1+\sqrt{4q+1}}{2}\right) \textrm{ b/u}.
\label{C_infinit}
\end{equation}
\label{Theorem2}
\end{thm}
\begin{proof}
Theorem \ref{Theorem2} is proved in the Appendix.
\end{proof}
\begin{rem} 
\begin{remenum}
\item $C_\infty(q)$ is achieved by the input pmf given in Table~\ref{p_Xv-1_Xv} where 
\begin{equation}
p_i = \frac{1}{2}\left( 1+\frac{1}{\sqrt{4q+1}} \right)
\label{p_i_Th2}
\end{equation}
for all $i \geq 1$. Another optimal input pmf is characterized by Table~\ref{p_X0_X1} and Table~\ref{p_Xv-1_Xv} for $i\geq 2$ when $p_i$ is replaced by (\ref{p_i_Th2}) for all $i \geq 1$. This is proved in the Appendix.
\item $C_\infty(1)=0.6942$~b/u is equal to the logarithm of the \textit{golden ratio}. Also remarkable, $C_\infty(2)$ is exactly $1$~b/u.
\item The maximum achievable rates with time-sharing and, thus, no timing are given by $R_{ts}(q)=0.5\log(q+1)$~b/u. For $q=1,2$ we have $0.5$ and $0.7925$~b/u, respectively. Since $C_\infty(q)$ is obviously a lower bound on the capacity of each finite length cascade, a comparison of the time-sharing rates with $C_\infty(1)$ and $C_\infty(2)$ shows that pre-determined time-sharing falls considerably short of the capacity for small transmission alphabets. For very large transmission alphabets the gap between the rates due to time-sharing and timing becomes negligible, i.e. $\lim_{q\rightarrow \infty} \left ( C_\infty(q)-R_{ts}(q)\right)=0$.
\end{remenum}
\end{rem}
Next we state an achievable rate region for a cascade with more than one source. Let $\mathcal{T}$ denote the set of all rate vectors $\mathbf{R}\in\mathbb{R}^m_+$ satisfying
{\setlength{\arraycolsep}{0.0em}
\begin{eqnarray}
\mathcal{T}&{}={}& \bigg\{\mathbf{R}\in\mathbb{R}_+^m:0 \leq \sum_{k=0}^{v}R_k \leq  H(Y_{v+1}|X_{v+1}), \nonumber\\
&&{\hspace{3.9cm}}\forall v\in \mathcal{V}\setminus \{m\}\bigg\}
\label{cut_set_bound_Lem2}
\end{eqnarray}}
which becomes, taking into account Table~\ref{p_Xv-1_Xv} and \ref{p_X0_X1}, (\ref{cut_set_bound_Lem2_explicit})
\begin{figure*}[!b]
\vspace*{4pt}
\hrulefill
\begin{equation}
\label{cut_set_bound_Lem2_explicit}
\mathcal{T}= \left\{\mathbf{R}\in\mathbb{R}_+^m: \left\{\begin{array}{ll} R_0 \leq p_1\log(q+1) &\\ \sum_{k=0}^{v}R_k \leq  \bar{p}_{v}\log q+p_{v+1}H\left(\bar{p}_v p_{v+1}^{-1}\right) & , \quad \forall v\in \mathcal{V}\setminus \{0,m\} \end{array}\right \}\right\}
\end{equation}
\end{figure*}
and let $\mathcal{U}$ denote the set of all $\mathbf{R}\in\mathbb{R}^m_+$ satisfying
{\setlength{\arraycolsep}{0.0em}
\begin{eqnarray}
\mathcal{U}&{}={}&  \bigg\{\mathbf{R}\in\mathbb{R}^m_+: 0< R_0 ,\hspace{0.1cm} 0 \leq R_{v}\leq \bar{p}_v\log q,\nonumber\\
&&{\hspace{3.5cm}}\forall v\in \mathcal{V}\setminus{\{0,m\}}\bigg\}
\label{Rate_Constraint}
\end{eqnarray}}
where $p_v\in(0,1)$, $p_{m}=1$ and $p_v + p_{v+1} \geq 1$.
\begin{thm}
Consider a noise-free relay cascade with $m-1$ half-duplex constrained relays where each relay can act as a source. The achievable rate region  $\mathcal{R}$ due to the timing strategy (see section~\ref{sec:coding_idea}) is given by
\begin{equation}
\mathcal{R} = \mathit{Co}\left (\bigcup_{p_{X_0\dots X_m}}\mathcal{T}\cap \mathcal{U}\right)
\label{Th:Rate_region}
\end{equation} 
where the union is over all assignments $p_{X_0\dots X_m}$ as shown in Table~\ref{p_Xv-1_Xv} and \ref{p_X0_X1}. 
\label{Theorem3}
\end{thm}
\begin{proof}
The input pmf stated in Table~\ref{p_Xv-1_Xv} and \ref{p_X0_X1} is still optimal for the case considered here. Taking into account the resulting entropy functions (\ref{H(X_0|X_1)}) and (\ref{H(X_v|X_{v+1})}), it follows that the achievable rate region~$\mathcal{R}$ due to (\ref{R_0}) to (\ref{Rate_region_rough}) equals (\ref{Th:Rate_region}) as $n \rightarrow \infty$.
\end{proof}
\begin{singlerem}
Observe that $\cup_{p_{X_0\dots X_m}}\mathcal{T}$ with $p_{X_0\dots X_m}$ as shown in Table~\ref{p_Xv-1_Xv} and \ref{p_X0_X1} is equal to the cut-set region (Lemma~\ref{Lemma2}). Thus, all boundary points of $\cup_{p_{X_0\dots X_m}}\mathcal{T}$ that are achievable when the constraints stated in (\ref{Rate_Constraint}) are satisfied are capacity points. This idea will be illustrated for an example in paragraph~\ref{Example_Two_Sources}.
\end{singlerem}

\section{Numerical Examples}\label{Section5}
In this section we shall provide numerical capacity results for various scenarios by means of Theorem~\ref{Theorem1} and Theorem~\ref{Theorem3}. In particular, we show how to obtain the capacity of a half-duplex constrained relay cascade with one source-destination pair for an arbitrary number of relays. Further, in case of a three node relay cascade with source and relay source, an explicit expression of the region due to Theorem~\ref{Theorem3} is derived. 

\subsection{One Source}
\label{Subsec_One_Source}
Let us first consider a relay cascade with $\mathcal{V}=\{0,1,2\}$, $q=2$ and $R_1=0$, i.e. source node $0$ intends to communicate with sink node $2$ via the half-duplex constrained relay $1$. By Theorem~\ref{Theorem1} and the optimum input pmf stated in Table~\ref{p_X0_X1}, we have
\begin{equation}
C_1(2)=\max_{p_{X_0X_1X_2}}\min\left \{p_{X_1}(\textrm{N})\log3,H(X_1)\right\}.
\label{C_single_rel}
\end{equation}
Problem (\ref{C_single_rel}) exhibits a single degree of freedom and is readily solved by finding a $p_{X_1}(\textrm{N})$ which satisfies $p_{X_1}(\textrm{N})\log3=H(X_1)$ (see Fig.~\ref{fig:C_single_relay}). 
\begin{figure*}[!t]
\centering
\psfrag{Pr}{\footnotesize$p_{X_1}(\textrm{N})$}
\psfrag{H(X1)}{\footnotesize$H(X_1)$}
\psfrag{pX1(N)log23space}{\footnotesize$p_{X_1}(\textrm{N})\log3$}
\psfrag{bits per transmission}{\footnotesize b/u}
\epsfig{file=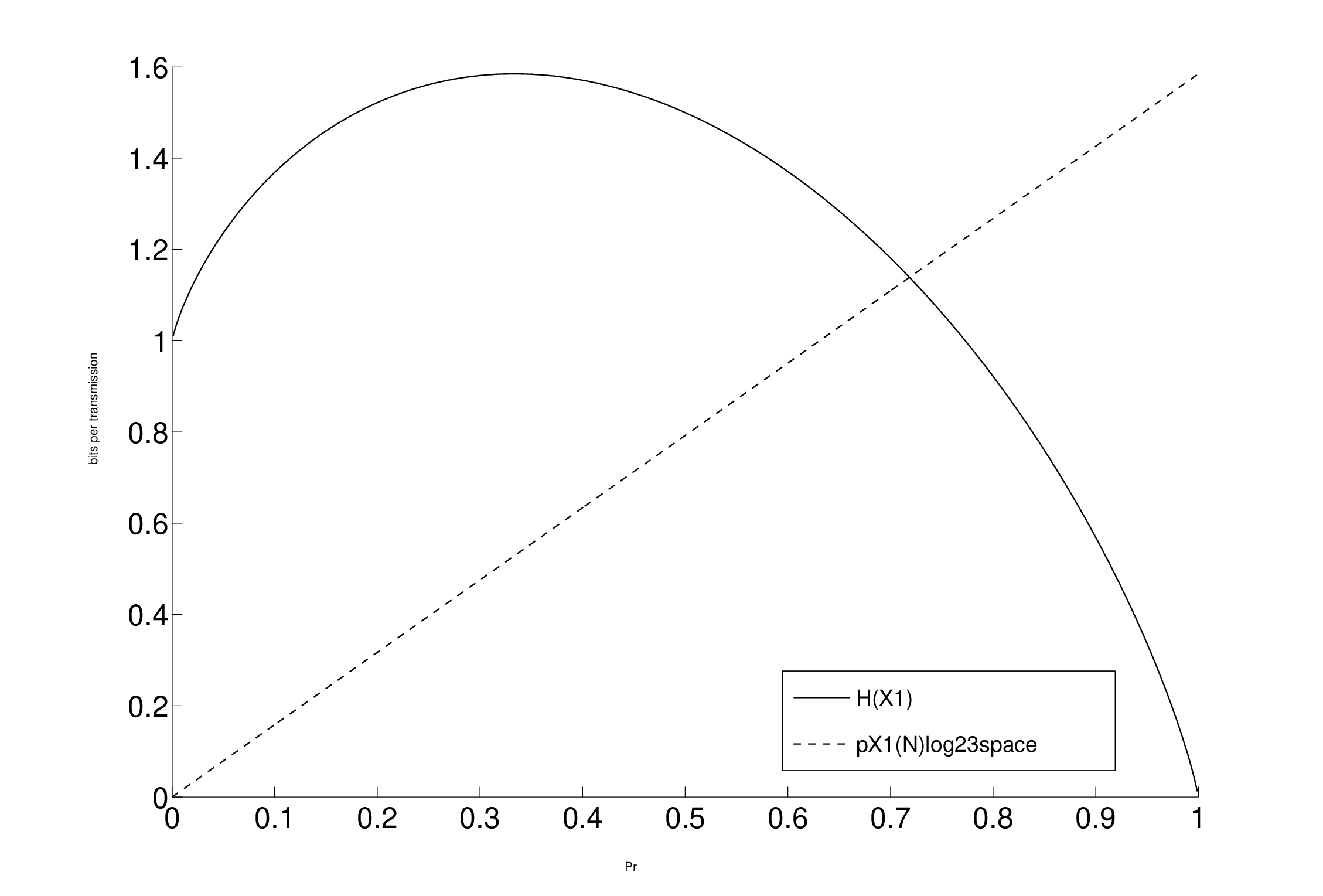,scale=0.3}
\caption{Graphical solution of optimization problem (\ref{C_single_rel}).}
\label{fig:C_single_relay}
\end{figure*}
The optimum value for $p_{X_1}(\textrm{N})$ equals $0.7185$ and results in 
\begin{equation}
C_1(2)=1.1389 \textrm{ b/u}.
\label{eq:1.1389}
\end{equation}
\begin{rem} 
\begin{remenum}
\item Assume the relay does not have the capability to decide whether the source has transmitted or not, i.e. $p_{X_0X_1}(\textrm{N},\textrm{N})=0$. In this case an identical approach shows that the capacity equals $0.8295$~b/u, which is still greater than the time-sharing rate of $0.5\log3\approx 0.7925$ bit per use.
\item For $q=1$, the outlined procedure yields $C_1(1)=0.7729$~b/u achieved by $p_{X_1}(\textrm{N})=0.7729$. The capacity value of this specific case has also been obtained in \cite{VanMeu92}. Therein, the focus was not on half-duplex constrained transmission but on finding the capacity of certain classes of deterministic relay channels. In \cite{Kra07}, the same channel model was considered and the author also noticed that the capacity equals $0.7729$~b/u. A simple coding scheme was outlined which approaches $2/3$~b/u, and extensions using Huffman or arithmetic source coding are claimed.
\end{remenum}
\end{rem}
In order to compute $C_{m-1}(q)$ for $m >2$, we transform (\ref{eq:Corollary1}) into a convex program with linear cost function $H(Y_1|X_1)$ and convex equality constraints $H(Y_1|X_1)-H(Y_{i+1}|X_{i+1}) = 0$ for all $i\in \{1,\dots,m-1\}$. The resulting program reads as
{
\begin{eqnarray}
\textrm{maximize}&& p_1\log(q+1) \nonumber \\
\label{Convex_Program}
\textrm{subject to}&& p_1\log(q+1)-\bar{p}_{i}\log q - p_{i+1}H\left(\bar{p}_i p_{i+1}^{-1}\right)=0\nonumber \\
&&1 - \sum_{j=i}^{i+1} p_j \leq 0 \nonumber\\
&&p_i \in (0,1) \nonumber
\end{eqnarray}
By adopting a standard algorithm for constrained optimization problems, the capacity $C_{m-1}(q)$ was computed for various values of $m$. A brief summary is given in Table~\ref{Tab:Numerical_C}. 
\begin{table}[!t]
\renewcommand{\arraystretch}{1.3}
\caption{Capacity results for cascades composed of $m-1$ half-duplex relays. Row ``TS'' shows the corresponding time-sharing rates}
\begin{center}
\begin{tabular}{c||c|c} 
\hline
$m-1$&$C_{m-1}(1)$&$C_{m-1}(2)$ \\ \hline \hline
$1$&$0.7729$ b/u&$1.1389$ b/u \\ 
$2$&$0.7324$ b/u&$1.0665$ b/u \\ 
$3$&$0.7173$ b/u&$1.0400$ b/u \\ 
$4$&$0.7099$ b/u&$1.0271$ b/u \\ 
$10$&$0.6981$ b/u&$1.0066$ b/u \\ 
$20$&$0.6954$ b/u&$1.0020$ b/u\\ 
$40$&$0.6946$ b/u&$1.0006$ b/u\\ 
$100$&$0.6943$ b/u&$1.0001$ b/u \\ 
$\infty$&$0.6942$ b/u&$1$ b/u \\ \hline\hline
TS&$0.5$ b/u&$0.7925$ b/u \\ \hline
\end{tabular}
\label{Tab:Numerical_C}
\end{center}
\end{table}

\subsection{Two Sources}\label{Example_Two_Sources}
The considered relay network is characterized by $\mathcal{V}=\{0,1,2\}$ and $q=2$. In contrast to the previous example, the relay is allowed to send own information, i.e. $R_1\geq 0$. According to Theorem~\ref{Theorem3}, the achievable rate region~$\mathcal{R}$ is given by the convex hull of
{\setlength{\arraycolsep}{0.0em}
\begin{eqnarray}
\label{Ex_two_sources:R_0}
0 \leq R_0 &{}\leq{}& H(X_0|X_1) \\
\label{Ex_two_sources:sumrate}
0 \leq R_0 + R_1 &{}\leq{}& H(X_1)\\
\label{Ex_two_sources:coding_constraint}
0 \leq R_1&{}\leq{}& \bar{p}_1\\
\label{Ex_two_sources:R_0>0}
R_0 &{}>{}&0
\end{eqnarray}}
together with $(R_0,R_1)=(0,\log3)$ which follows by considering the shortened cascade from the relay to the source. Observe that (\ref{Ex_two_sources:R_0}) and (\ref{Ex_two_sources:sumrate}) correspond to $\mathcal{T}$ while (\ref{Ex_two_sources:coding_constraint}) and (\ref{Ex_two_sources:R_0>0}) correspond to $\mathcal{U}$. 

We will first derive an explicit expression for the boundary of the cut-set region~$\mathcal{T}$. Two cases have to be considered depending on whether an optimum input pmf for the source or the relay source is used. An optimum input pmf for the relay source due to Table~\ref{p_Xv-1_Xv} is shown in Table~\ref{Ex_two_sources:opt_p_relay_source}. It yields the maximum possible sum rate $H(X_1)=\log3$~b/u for all  valid $y$ (i.e. $y\in[0,1/6]$). When $y$ varies from $0$ to $1/6$, we have $0 \leq H(X_0|X_1)\leq \frac{1}{3}\log3$ where $\frac{1}{3}\log3$ corresponds to $y=1/9$. Thus, a part of the cut-set region boundary is given by $R_1=\log3-R_0$ for $0\leq R_0 \leq \frac{1}{3}\log3$.
\begin{table}[th]
\renewcommand{\arraystretch}{1.3}
\caption{Optimal $p_{X_0X_1}$ yielding a sum rate of $\log3$ $b/u$}
\begin{center}
\begin{tabular}{c|c c c}
\multicolumn{1}{c}{\backslashbox{\normalsize{$X_0$}}{\normalsize{$X_1$}}} &$0$ &$1$& N \\ \cline{2-4}
$0$& $0$ & $0$ & $y$\\ 
$1$& $0$ & $0$ & $y$\\ 
N&$1/3$& $1/3$& $1/3-2y$\\ 
\end{tabular}
\label{Ex_two_sources:opt_p_relay_source}
\end{center}
\end{table}
It remains to focus on the interval $\frac{1}{3}\log3 < R_0 \leq C_1(2)=1.1389$~b/u. Using the optimum input pmf for source node $0$ (Table~\ref{p_X0_X1}) and (\ref{H(X_0|X_1)}), we can express $R_1=H(X_1)-R_0$ as shown in (\ref{Ex_two_sources:cut_set_bound_b}). Hence, the boundary of the cut-set region is given by (\ref{Ex_two_sources:cut_set_bound})
\begin{figure*}[!b]
\vspace*{4pt}
\hrulefill
\begin{subnumcases}{\label{Ex_two_sources:cut_set_bound}R_1 = }
\log 3 - R_0 & $0\leq R_0 \leq \frac{1}{3}\log 3\textrm{ b/u}$\label{Ex_two_sources:cut_set_bound_a}\\
H\left( \frac{R_0}{\log 3}\right) + \left(1- \frac{R_0}{\log 3}\right) -R_0 & $\frac{1}{3}\log 3 < R_0 \leq 1.1389\textrm{ b/u}$\label{Ex_two_sources:cut_set_bound_b}
\end{subnumcases}
\end{figure*}
In order to determine $\mathcal{R}$, (\ref{Ex_two_sources:coding_constraint}) must be taken into account. We first check whether points on (\ref{Ex_two_sources:cut_set_bound_a}) are achievable under constraint (\ref{Ex_two_sources:coding_constraint}). Using the probability mass function of Table~\ref{Ex_two_sources:opt_p_relay_source}, it follows from (\ref{Ex_two_sources:coding_constraint}) that $R_1 \leq \frac{2}{3}$~b/u. Hence, no point (except of $(0,\log3)$) is achievable on (\ref{Ex_two_sources:cut_set_bound_a}) since the range of $R_0$ implies that $R_1$ is always greater or equal $\frac{2}{3}\log3$~b/u. Let us now focus on (\ref{Ex_two_sources:cut_set_bound_b}) and recall that Table~\ref{p_X0_X1} is the underlying probability function. Rate points on (\ref{Ex_two_sources:cut_set_bound_b}) which satisfy
\begin{equation}
H(X_0|X_1)+\bar{p}_1 \leq H(X_1).
\label{Ex_two_sources:lower_bound_R0}
\end{equation}
are achievable. Equality in (\ref{Ex_two_sources:lower_bound_R0}) results for $p_1=0.6091$ which gives $R_0=0.9654$~b/u and $R_1=0.3909$~b/u. Since $H(X_0|X_1)+\bar{p}_1$ is linear in $p_1$ while $H(X_1)$ is concave in $p_1$, (\ref{Ex_two_sources:lower_bound_R0}) is satisfied for all $p_1\leq0.6091$. The corresponding rate points are $R_0\geq0.9654$ and $R_1\leq0.3909$~b/u. Thus, $\mathcal{R}$ is given by taking the convex hull of (\ref{Ex_two_sources:cut_set_bound_b}) for $0.9654 \leq R_0 \leq 1.1389$~b/u and the rate vector $(R_0,R_1)=(0,\log3)$~b/u. The cut-set bound, the timing region $\mathcal{R}$ and the region which results from a deterministic time-division schedule (i.e. time-sharing between $(R_0,R_1)=(0.5\log3,0)$ and $(0,\log3)$) is depicted in Fig.~\ref{Fig:Rate_region_two_sources_binary}. 

The derivation reveals that the cut-set bound is achievable for $R_0\geq0.9654$. Moreover, we see that even when the source transmits at a rate beyond the time-sharing rate of $0.5\log3$~b/u, the relay is still able to send its own information at a non-zero rate.
\begin{figure*}[!t]
\psfrag{(a)}{\footnotesize{(a)}}
\psfrag{(b)}{\footnotesize{(b)}}
\psfrag{(c)}{\footnotesize{(c)}}
\psfrag{R0}{\footnotesize{\textrm{$R_0$ (b/u)}}}
\psfrag{R1}{\footnotesize{\textrm{$R_1$ (b/u)}}}
\centering
\epsfig{file=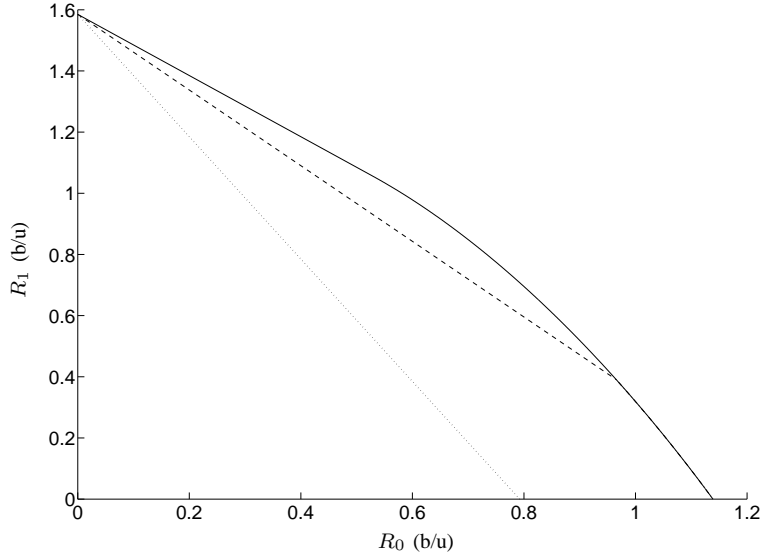, scale = 0.3}
\caption{The cut-set region is bounded by the solid curve, $\mathcal{R}$ is bounded by the dashed line and the the solid curve for $R_0 \geq 0.9654$~b/u. The time-sharing region is bounded by the dotted line.}
\label{Fig:Rate_region_two_sources_binary}
\end{figure*}

\section{Extension to other Networks}
Relay cascades are fundamental building blocks in communication networks. The results derived in the previous sections may be instrumental in order to determine the capacity of half-duplex constrained networks with more elaborate topologies. 
\subsection{Wireless Trees}
Consider, for instance, the tree structured network depicted in Fig.~\ref{Fig:Wireless_Tree}. The root (node~$1$) wants to multicast information to all leaves (nodes $2$ to $8$) via four half-duplex constrained relays. We assume noise-free bit pipes (i.e. $q=1$) and broadcast behavior at nodes with more than one outgoing arrow. The multicast capacity is limited by the capacity of the longest path in the tree which goes from node $1$ to nodes $7$ and $8$. Hence, the multicast capacity in the considered example is equal to the capacity of a cascade containing two intermediate relay nodes, which is $C_2(1)=0.7324$~b/u (see Table~\ref{Tab:Numerical_C}).
\begin{figure}[ht]
\psfrag{1}{$1$}
\psfrag{2}{$2$}
\psfrag{3}{$3$}
\psfrag{4}{$4$}
\psfrag{5}{$5$}
\psfrag{6}{$6$}
\psfrag{7}{$7$}
\psfrag{8}{$8$}
\centering
\epsfig{file=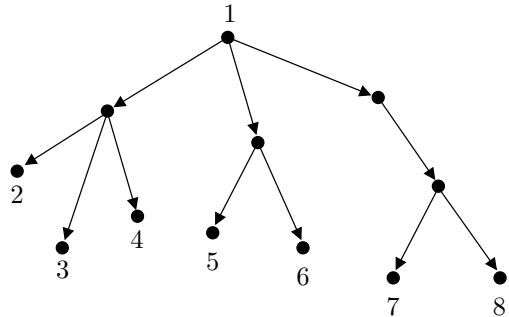, scale = 1}
\caption{A wireless binary tree. The multicast capacity is equal to $C_2(1)=0.7324$~b/u.}
\label{Fig:Wireless_Tree}
\end{figure}

\subsection{The Half-Duplex Butterfly Network}
A half-duplex butterfly network~\cite{AhlCaiLiYeu00} is shown in Fig.~\ref{Fig:Wireless_Butterfly}. Nodes~$1$ and $2$ intend to multicast information to sink nodes $4$ and $5$ via both a direct link and a half-duplex constrained relay node~$3$. Like before, broadcast transmission and bit pipes are assumed. All nodes with two incoming arrows behave according to a collision model, i.e. received information is erased if there was a transmission on both incoming links. By means of network coding (NC) with a bit-wise XOR, $\frac{2}{3}$~b/u are achievable at the sink nodes. The (well-known) strategy is (see Fig.~\ref{Fig:Wireless_Butterfly}) to send in the first time slot a binary symbol $u_1$ via broadcast~$a$ to nodes~$3$ and $4$, in the second time slot a binary symbol $u_2$ via broadcast~$b$ to nodes~$3$ and $5$ and, subsequently, in the third time slot $u_1\oplus u_2$ via broadcast~$c$ from the relay node to both sinks. However, under the usage of timing, at least $0.7729$~b/u is achievable by applying the proposed timing strategy as follows. Information originating from node~$1$ can be sent by means of timing at a rate of $C_1(1)=0.7729$~b/u concurrently on paths $1,1\times 4,4$ and $1,1\times 3,3,3\times 5,5$. Similarly, information originating from node~$2$ can be sent by means of timing at a rate of $C_1(1)=0.7729$~b/u concurrently on paths $2,2\times 3,3,3\times 4,4$ and $2,2\times 5,5$. Hence, time-sharing of both source nodes yields a multicast rate of $0.7729$~b/u. Assume for the moment that node~$1$ is sending information. Decoding at sink nodes~$4$ and $5$ is done as follows. First observe that the sequence received at sink node $4$ is a superposition of the sequence sent by source node~$1$ on the direct link  $1\times 4$ and of the relay sequence on $3\times 4$. Due to the timing-strategy source node~$1$ and the relay never transmit in the same time slot. Hence, sink node~$4$ is able to extract the information sent by source node~$0$ from the received sequence by the following protocol. In the very first block source node~$0$ forwards a message to sink node~$4$ and the relay via broadcast~$a$ while the relay is quiet. Sink node~$4$ and the relay are able to decode successfully. In the second block the relay sends the decoded message to nodes~$4$ and $5$ via broadcast~$c$ while source node~$0$ sends a new message to nodes~$3$ and $4$ via broadcast~$a$. Since sink node~$4$ knows both the strategy and, therefore, the current sequence used by the relay for encoding the source message of the previous block, it can determine the new source message by subtracting the relay sequence from the received sequence. Sink node~$5$ is also able to decode the received relay sequence by applying the rules for the proposed timing strategy. The outlined procedure is repeated in the following blocks and is used in the same way for transmitting information from source node~$2$ to nodes~$4$ and $5$.
\begin{figure}[htbp]
\psfrag{a}{\small{$a$}}
\psfrag{b}{\small{$b$}}
\psfrag{c}{\small{$c$}}
\psfrag{d}{\small{$d$}}
\psfrag{1}{$1$}
\psfrag{2}{$2$}
\psfrag{3}{$3$}
\psfrag{4}{$4$}
\psfrag{5}{$5$}
\centering
{\epsfig{file=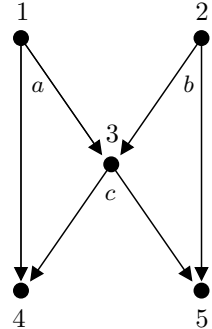, scale = 1.2}} 
\caption{The binary half-duplex butterfly network. With network coding, $\frac{2}{3}$~b/u are achievable. Timing yields $C_1(1)=0.7729$~b/u.}
\label{Fig:Wireless_Butterfly}
\end{figure}

\section{Conclusion}
The half-duplex constraint is a property common to many wireless networks. In order to overcome the half-duplex constraint, practical transmission protocols deterministically split the time of each network node into transmission and reception periods. However, this is not optimal from an information theoretic point of view, as is demonstrated by means of noise-free relay cascades of various lengths with one or multiple sources. We show that significant rate gains are possible when information is represented by an information-dependent allocation of the transmission and reception slots of the relays. Moreover, we provide a coding strategy which realizes this idea and, based on the asymptotic behavior of the strategy, we establish capacity expressions for three different scenarios. These results may be instrumental in deriving the capacity of half-duplex constrained networks with a more elaborate topology.

\appendix[]\label{Appendix}
\begin{lem}
Consider a noise-free relay cascade with a single source-destination pair (namely nodes $0$ and $m$) and $m-1$ half-duplex constrained relays where $q$ denotes the number of transmission symbols. There exists a capacity achieving input pmf $p_{X_0\dots X_{m}}$ such that $C_{m-1}(q)=H(X_{m-1})$.
\label{Lemma_App}
\end{lem}
\begin{IEEEproof}
Consider the capacity expression of Theorem~\ref{Theorem1} and assume that $H(Y_{m}|X_m)>C_{m-1}(q)$. It will be shown that $H(Y_{m}|X_m)$ can be decreased to $C_{m-1}(q)$ without forcing $H(Y_i|X_i)$, $1\leq i \leq m-1$, to decrease. The optimal input pmf given in Table~\ref{p_Xv-1_Xv} and \ref{p_X0_X1} is assumed in the following. Hence, $H(Y_m|X_m)=H(X_{m-1}|X_m)=H(X_{m-1})$ and $H(Y_i|X_i)=H(X_{i-1}|X_i)$ for all $1\leq i \leq m-1$. The assertion is clear for $m=2$ (see Fig.~\ref{fig:C_single_relay}). Let~$m>2$. Recall that $H(X_0|X_1)=p_1\log(q+1)$ and 
\begin{equation}
H(X_{i-1}|X_i)= \bar{p}_{i-1}\log q + p_iH\left(\bar{p}_{i-1}p_i^{-1}\right)
\label{cond_entr_App}
\end{equation}
where $2\leq i \leq m$ and $p_m=1$. A change of $H(X_{m-1})$ does not affect $H(X_{i-1}|X_i)$, $1\leq i \leq m-2$, since both expressions depend on different variables. Therefore, it is enough to consider $H(X_{m-2}|X_{m-1})$. The maximum of $H(X_{m-1})$ is at $p_{m-1}=1/(q+1)$. Further, $H(X_{m-1})$ is (strictly) decreasing to zero for $1/(q+1)\leq p_{m-1} \leq 1$. Let $p_{m-1} \geq 1/(q+1)$. In order to decrease $H(X_{m-1})$ to $C_{m-1}(q)$, $p_{m-1}$ has to be increased which, in turn, does not decrease $H(X_{m-2}|X_{m-1})$ since
\begin{equation}
\frac{\partial H(X_{m-2}|X_{m-1})}{\partial p_{m-1}}=\log\left( \frac{p_{m-1}}{p_{m-1}+p_{m-2}-1}\right).
\label{part_deriv_Lem}
\end{equation}
is non-negative. The assertion is proven since we do not have to consider $p_{m-1} < 1/(q+1)$. Such a choice would only decrease $H(X_{m-2}|X_{m-1})$ but would not result in larger values for $H(X_{m-1})$.
\end{IEEEproof}

\begin{IEEEproof}[Proof of Theorem~\ref{Theorem2}]
The capacity series $(C_m(q))_{m\in \mathbb{N}}$ is bounded (e.~g. by $0$ and $C_1(q)$) and monotonically decreasing (since each new relay causes an additional constraint in the corresponding convex program of section~\ref{Subsec_One_Source}). Hence, $(C_m(q))_{m\in \mathbb{N}}$ is convergent. Thus, for every $\epsilon >0$ there exists an $N \in \mathbb{N}$ such that
\begin{equation}
|C_{m-1}(q)-C_{m}(q)|< \epsilon
\label{Cauchy_condition}
\end{equation}
for all $m \geq N$. Assuming the capacity achieving input pmf, we have $C_{m-1}(q)=H(X_{m-1})$ and $C_{m}(q)=H(X_{m})$ (Lemma~\ref{Lemma_App}). Then, by (\ref{Cauchy_condition})
\begin{equation}
|H(X_{m-1})-H(X_{m})|< \epsilon
\label{Cauchy_condition_2}
\end{equation}
for all $m \geq N$. Two cases can appear in (\ref{Cauchy_condition_2}) when $\epsilon$ approaches zero: $p_{m-1}\rightarrow p^\prime$, $p_{m}\rightarrow p^{\prime\prime}$  as $m\rightarrow \infty$ with  $p^\prime\neq p^{\prime\prime}$ or $p^\prime = p^{\prime\prime}$.

Consider the first case, i.e. $p^\prime\neq p^{\prime\prime}$. By (\ref{constraint_on_pi_pi+1}), $p^\prime + p^{\prime\prime}$ has to be greater than or equal to $1$. However, $p^\prime + p^{\prime\prime}$ is always smaller than $1$ what can be seen as follows. First, note that the maximum of $H(X_k)$ is at $1/(q+1)$. Hence, without restriction we can assume that $p^\prime < 0.5$ and $p^{\prime\prime}>0.5$ (otherwise $p^\prime + p^{\prime\prime}<1$ a priori). Since the first derivative of $H(X_k)$ is point symmetric with respect to $(0.5,-\log q)$, we have $0.5-p^{\prime}>p^{\prime \prime}-0.5$ what yields $p^\prime + p^{\prime\prime}<1$. 

Hence, only the second case is valid, i.e. $p_{m-1}, p_{m}\rightarrow p$ as $m\rightarrow \infty$. But this implies, using Table~\ref{p_Xv-1_Xv} and replacing $p_{m-1}$ and $p_m$ by $p$, that $C_\infty(q)$ is smaller than or equal to the maximum of
\begin{equation}
H(X_{m-1}|X_m)= \bar{p}\log q + pH\left(\bar{p}p^{-1}\right).
\label{entropy_to_max_App}
\end{equation}
Since the Table~\ref{p_Xv-1_Xv} with $p_i-1$ and $p_i$ can be assumed for all $p_{X_{i-1}X_i}$, $i \geq 1$, it follows that $C_\infty(q)$ is equal to the maximum of (\ref{entropy_to_max_App}) which is 
\begin{equation}
\max_p H(X_{m-1}|X_m)=\log\left( \frac{1+\sqrt{4q+1}}{2} \right).
\label{App:C_infinity}
\end{equation}
\end{IEEEproof}
\begin{IEEEproof}[Proof of Theorem~\ref{Theorem2} - Remark i)]
The maximum of (\ref{App:C_infinity}) is achieved at 
\begin{equation}
p=\frac{1}{2}\left ( 1+\frac{1}{\sqrt{4q+1}}\right).
\label{Proof_Th2:listen_fraction}
\end{equation}
and an optimal input pmf is given by Table~\ref{p_Xv-1_Xv} when $p_{i}$ is replaced by (\ref{Proof_Th2:listen_fraction}) for all $i\geq 1$. Note that $p_{X_0X_1}$ is also characterized by Table~\ref{p_Xv-1_Xv} . Another optimal input pmf is given by Table~\ref{p_Xv-1_Xv} and Table~\ref{p_X0_X1} when $p_i$ is replaced by (\ref{Proof_Th2:listen_fraction}) for all $i\geq 1$. Since $p_{X_0X_1}$ is the only part which differs from the pmf considered before, it suffices to show that the value of $H(X_0|X_1)$ under the claimed pmf is always greater or equal to (\ref{App:C_infinity}), i.e. 
\begin{equation}
\frac{1}{2}\left ( 1+\frac{1}{\sqrt{4q+1}}\right)\log(q+1)\geq\log\left( \frac{1+\sqrt{4q+1}}{2} \right)
\label{App_first_step}
\end{equation}
or, equivalently,
\begin{equation}
(q+1)^{\frac{1}{2}\left ( 1+\frac{1}{\sqrt{4q+1}}\right)}\geq\frac{1+\sqrt{4q+1}}{2}.
\label{App_last_step_start}
\end{equation}
Lowering the left hand side while increasing the right hand side gives
\begin{equation}
(q+1)^{\frac{1}{2}\left ( 1+\frac{1}{2\sqrt{q+1}}\right)}\geq\frac{1+2\sqrt{q+1}}{2}.
\label{App_last_step}
\end{equation}
Using the substitution
\begin{equation}
\tilde{q}=\frac{1}{2\sqrt{q+1}}
\end{equation}
in (\ref{App_last_step}), we obtain
\begin{equation}
(2\tilde{q})^{-\tilde{q}}\geq \tilde{q}+1.
\label{App_simplified_equation}
\end{equation}
(\ref{App_simplified_equation}) is satisfied for all $\tilde{q}\in [0,0.2]$ what can be seen as follows. First note that (\ref{App_simplified_equation}) is satisfied for $\tilde{q}=0$ and $\tilde{q}=0.2$. Since $(2\tilde{q})^{-\tilde{q}}$ is concave due to a non-positive second derivative in the considered domain, (\ref{App_simplified_equation}) is valid for all $\tilde{q}\in [0,0.2]$. Thus, (\ref{App_first_step}) is true for all $q > 5$. The validity of (\ref{App_first_step}) for the remaining $q\in \{1,\dots,5\}$ is easily checked by direct computation.
\end{IEEEproof}

\section*{Acknowledgment}
We would like to thank Prof. Michelle Effros and Prof. Gerhard Kramer for carefully reading the manuscript and for providing many useful suggestions. Further, we would like to thank the anonymous reviewers and the associate editor for helpful comments.

\bibliographystyle{unsrt}

\enlargethispage{-2.7cm}
\begin{IEEEbiographynophoto}{Tobias Lutz}
was born in Krumbach, Germany, on May 02, 1980. He received the B.Sc. and Dipl.-Ing. degree in Electrical Engineering and Information Technology from the Technische Universit\"{a}t M\"{u}nchen, Germany in 2007 and 2008, respectively. Funded by a grant from the American European Engineering Exchange program he studied Electrical and Computer Engineering at the Rensselaer Polytechnic Institute, Troy, NY, from 2004 to 2005. Since 2008, he has been with the Institute for Communications Engineering at the Technische Universit\"{a}t M\"{u}nchen where he is currently working towards the Dr.-Ing. degree. Since 2009, he has been studying financial mathematics at the Ludwig Maximilian Universit\"{a}t M\"{u}nchen. His current research interests include information and coding theory and their application to wireless relay networks.
\end{IEEEbiographynophoto}
\begin{IEEEbiographynophoto}{Christoph Hausl}
(S'05) received the Dipl.-Ing. and Dr.-Ing. degree in Electrical Engineering and Information Technology from the Technische Universit\"{a}t M\"{u}nchen, Germany in 2004 and 2008, respectively. Since 2004, he has been with the Institute for Communications Engineering at the Technische Universit\"{a}t M\"{u}nchen as a research and teaching assistant. He is co-recipient of best paper awards at the International Conference on Communications (ICC) 2006 and at the International Workshop on Wireless Ad-hoc and Sensor Networks (IWWAN) 2006. He served as guest editor for the special issue on Physical Layer Network Coding for Wireless Cooperative Networks of the EURASIP Journal on Wireless Communications and Networking in 2010. His current research interest include channel coding and network coding and their application to wireless relay networks and mobile communications.
\end{IEEEbiographynophoto}
\begin{IEEEbiographynophoto}{Ralf K\"{o}tter}
(S'91--M'96--SM'06--F'09) was born in K\"{o}nigstein im Taunus, Germany, on October 10, 1963. He received a Diploma in Electrical Engineering from the Technische Universit\"{a}t Darmstadt, Germany, in 1990 and the Ph.D. degree from the Department of Electrical Engineering, Link\"{o}ping University, Sweden, in 1996. From 1996 to 1997, he was a Visiting Scientist at the IBM Almaden Research Center, San Jose, CA. He was a Visiting Assistant Professor at the University of Illinois, Urbana-Champaign, and a Visiting Scientist at CNRS in Sophia-Antipolis, France, from 1997 to 1998. During 1999--2006, he was member of the faculty at the University of Illinois, Urbana-Champaign. In 2006, he joined the faculty of the Technische Universit\"{a}t M\"{u}nchen, Germany, as the Head of the Institute for Communications Engineering (Lehrstuhl f\"{u}r Nachrichtentechnik). He passed away on February 2, 2009.

His research interests were in coding theory and information theory, and in their applications to communication systems.

During 1999--2001, Prof. K\"{o}tter was an Associate Editor for Coding Theory and Techniques for the IEEE TRANSACTIONS ON COMMUNICATIONS, and during 2000--2003, he served as an Associate Editor for Coding Theory for the IEEE TRANSACTIONS ON INFORMATION THEORY. He was Technical Program Co-Chair for the 2008 International Symposium on Information Theory, and twice Co-Editor-in-Chief for special issues of the IEEE TRANSACTIONS ON INFORMATION THEORY. During 2003--2008, he was a member of the Board of Governors of the IEEE Information Theory Society. He received an IBM Invention Achievement Award in 1997, an NSF CAREER Award in 2000, an IBM Partnership Award in 2001, and a Xerox Award for faculty research in 2006. He also received the IEEE Information Theory Society Paper Award in 2004, the Vodafone Innovationspreis in 2008, the Best Paper Award from the IEEE Signal Processing Society in 2008, and the IEEE Communications Society \& Information Theory Society Joint Paper Award twice, in 2009 and in 2010.
\end{IEEEbiographynophoto}
\end{document}